\documentclass[aip,jmp,reprint,onecolumn]{revtex4-1}

\usepackage{amsmath,empheq,braket,amsthm,amssymb,graphicx,amsfonts,array,color,subfigure,mathrsfs}

\usepackage[utf8]{inputenc}
\usepackage[T1]{fontenc}
\usepackage{mathptmx}
\usepackage{etoolbox}
\usepackage{bm}% bold math
\usepackage{scalerel} % For majorization type symbols
\usepackage{dsfont} % For identity
\usepackage[normalem]{ulem}
\usepackage[unicode=true,pdfusetitle, bookmarks=true,bookmarksnumbered=false,bookmarksopen=false,breaklinks=false,pdfborder={0 0 0},backref=false]{hyperref}

\DeclareMathOperator{\per}{per}

\newtheorem{theo}{Theorem}
\newtheorem{lem}{Lemma}
\newtheorem{cor}{Corollary}

\makeatletter
\def\@email#1#2{%
 \endgroup
 \patchcmd{\titleblock@produce}
  {\frontmatter@RRAPformat}
  {\frontmatter@RRAPformat{\produce@RRAP{*#1\href{mailto:#2}{#2}}}\frontmatter@RRAPformat}
  {}{}
}%
\makeatother

\begin{document}

\title{Boson--fermion complementarity in a linear interferometer:\\ an identity relating the determinant and permanent of a matrix}

\author{Michael~G. Jabbour}
\affiliation{SAMOVAR, Télécom SudParis, Institut Polytechnique de Paris, 91120 Palaiseau, France}
\affiliation{Centre for Quantum Information and Communication, \'Ecole polytechnique de Bruxelles, CP 165/59, Universit\'e libre de Bruxelles, 1050 Brussels, Belgium}
\author{Nicolas~J. Cerf}
\affiliation{Centre for Quantum Information and Communication, \'Ecole polytechnique de Bruxelles, CP 165/59, Universit\'e libre de Bruxelles, 1050 Brussels, Belgium}

\begin{abstract}
    Bosonic and fermionic statistics are well known to give rise to antinomic behaviors, most notably boson bunching \textit{vs} fermion antibunching. Here, we establish a fundamental relation that combines bosonic and fermionic multiparticle interferences in an arbitrary linear interferometer. The bosonic and fermionic transition probabilities appear together in a single equation which constrains their values, hence expressing a boson--fermion complementarity that is independent of the details of the interferometer. For two particles in any interferometer, for example, it implies that the average between the bosonic and fermionic probabilities must coincide with the probability obeyed by classical particles. Crucially, this fundamental relation also provides a heretofore unknown mathematical identity connecting the squared moduli of the permanent and determinant of an arbitrary complex matrix, hence extending an identity by Muir dating from the nineteenth century.
\end{abstract}

\maketitle

\section{Introduction}

Quantum interference is responsible for some of the most counterintuitive phenomena allowed by the laws of physics. Crucially, the way identical particles interfere depends primarily on their statistics. On the one hand, the symmetrization of the bosonic many-body wavefunction favors the so-called bunching of identical bosons~\cite{Einstein}, as witnessed for instance by the celebrated Hong-Ou-Mandel effect~\cite{HOM} or Hanbury Brown and Twiss effect~\cite{HBTwiss,Kimble}. On the other hand, identical fermions have a tendency to antibunch as a consequence of the antisymmetrization of the fermionic many-body wavefunction dictated by the Pauli principle~\cite{Pauli}. This antibunching of fermions has more recently also been demonstrated through the Hanbury Brown and Twiss effect~\cite{Henny1999,Kiesel2002,Rom2006,Jeltes2007}.
\par

In addition to its value from the point of view of fundamental physics, quantum interference is a key resource in many technologies currently rapidly gaining broad interest, such as quantum computing~\cite{Ladd2010}, quantum cryptography~\cite{Gisin2002}, and
superconducting quantum interference devices~\cite{Vasyukov2013}. Strikingly, it turns out to be a crucial aspect of the so-called boson sampling paradigm~\cite{bosonsampling}, which focuses on the computational complexity of simulating the scattering of many identical bosons through a multimode linear interferometer. Boson sampling is generally regarded as an instance of a classically hard computational problem that can be efficiently solved by a quantum computer~\cite{Tillmann2013,Crespi2013,Broome2013,Carolan2014}. Interestingly, a similar conclusion has recently been reached for fermion sampling~\cite{Oszmaniec2022}, in which fermionic linear optics circuits fed with specifically chosen entangled input states were used to demonstrate a quantum computational advantage.
\par

The quantum interference of multiple identical  particles has been extensively studied since the discovery of the Hong-Ou-Mandel effect. More recently, there has been an attempt to compare the behavior of multiparticle interferences in terms of the statistics of the particles. While two fermions will always antibunch and two bosons will favor bunching, it has for instance been shown that the collective interference of more particles encompass much more diverse behaviors~\cite{Tichy2012a}. Recent studies have also focused on more general systems such as the interference between composite two-fermion bosons~\cite{Tichy2012b}, the study of many-body states that obey non-pairwise exchange symmetries~\cite{Tichy2017}, and partially distinguishable bosons and fermions~\cite{Spivak2022}.
\par

Multiparticle quantum interferences are notoriously difficult to characterize.
The transition probabilities governing such interferences can be expressed in terms of the permanent or determinant of a specific matrix, depending on whether we deal with bosons [see Eq.~\eqref{eq:defBos}] or fermions [see Eq.~\eqref{eq:defFer}]. This usually leads to cumbersome expressions which can be understood as originating from the many interfering paths accessible to the different particles. These expressions can of course be evaluated, but cannot easily be exploited analytically.
\par

Here, we establish a fundamental relation that governs bosonic and fermionic multiparticle interferences in an arbitrary linear interferometer. This relation is amazingly simple, considering the complexity of quantum interferences. To our knowledge, it is moreover the first time that bosonic and fermionic transition probabilities are combined in a single unifying equation. By constraining the values of the transition probabilities for both types of particles, this relation thus expresses some boson \textit{vs} fermion complementarity. This relation solely depends on the duality between symmetric and antisymmetric statistics and encompasses all linear couplings between the particles (i.e., couplings that preserve the total particle number). To prove our main result, we make use of the powerful mathematical tool of the generating function which, for bosons, enables an analytical calculation exploiting the Gaussian formalism as previously remarked in Ref.~\onlinecite{Jabbour-Cerf-PRR-2021} (see also Ref.~\onlinecite{Cerf2020}).
\par

The rest of the paper is organized as follows. We begin with a preliminary section where we briefly introduce the theory of bosonic and fermionic linear interferometers. We then present our main results. Theorem~\ref{theo:RecNmain} exhibits a novel relation that governs the boson--fermion complementarity. Theorem~\ref{theo:recNGen} then provides a fundamental identity connecting the squared moduli of the permanent and determinant of an arbitrary complex matrix. Next, we elaborate on the physical interpretation of the above complementarity. This is followed with a section detailing the proofs of our two main Theorems. Finally, we finish with a conclusion and some open problems.
\par

%%%%%%%%%%%%%%%%%%%%%%%%
%%%%%%%%%%%%%%%%%%%%%%%%
%%%%%%%% PRELIM %%%%%%%%
%%%%%%%%%%%%%%%%%%%%%%%%
%%%%%%%%%%%%%%%%%%%%%%%%

\section{Preliminaries}

An $N$-mode bosonic system~\cite{Weedbrook2012} is described by the tensor product of $N$ infinite-dimensional Hilbert spaces, each characterized by a pair of bosonic field operators $\hat{a}_s, \hat{a}_s^{\dagger}$, where $s \in [N]\coloneqq \left\lbrace 1, \cdots ,N\right\rbrace$ denotes the mode index. The field operators satisfy the bosonic commutation relations $[\hat{a}_s,\hat{a}_{s'}^\dag] = \delta_{s{s'}} \mathds{1}$ (where $\mathds{1}$ is the identity in infinite dimension),  $[\hat{a}_s,\hat{a}_{s'}] = [\hat{a}_s^\dag,\hat{a}_{s'}^\dag] = 0$. In contrast, an $N$-mode fermionic system~\cite{Szalay2021} is described by the tensor product of $N$ two-dimensional Hilbert spaces, with the other distinction that each mode is now associated with a pair of fermionic field operators $\hat{b}_s, \hat{b}_s^{\dagger}$ that satisfy the fermionic anticommutation relations $\{\hat{b}_s,\hat{b}_{s'}^\dag\} = \delta_{s{s'}} \mathds{1}$, $\{\hat{b}_s,\hat{b}_{s'}\} = \{\hat{b}_s^\dag,\hat{b}_{s'}^\dag\} = 0$. The quantum states of both systems can be written using their respective Fock spaces, whose basis states we denote as $\ket{\boldsymbol{i}}\coloneqq \prod\nolimits_{s=1}^N \ket{i_s}$ with $\boldsymbol{i} \in \mathbb{N}^N$ for bosons, where $\mathbb{N}$ is the set of natural numbers (including zero), or $\boldsymbol{i} \in \{0,1\}^N$ for fermions.

Consider now an $N$-mode linear interferometer whose effect on the mode operators is characterized in phase space by a unitary matrix $U$ of dimension $N$,
{which acts on the vector of field operators $\boldsymbol{\hat{c}} \coloneqq \left(\hat{c}_1, \cdots, \hat{c}_N \right)^T$ as follows:
\begin{equation}
    \boldsymbol{\hat{c}} \rightarrow U \boldsymbol{\hat{c}} =
    \left( \sum_{s=1}^N u_{1s} \hat{c}_s, \sum_{s=1}^N u_{2s} \hat{c}_s, \hdots, \sum_{s=1}^N u_{Ns} \hat{c}_s \right)^T,
\end{equation}
where $u_{ks}$ is the component of $U$ on the $k$\textsuperscript{th} row and the $s$\textsuperscript{th} column, while}
$\hat{c}_s = \hat{a}_s$ (or $ \hat{b}_s$) for a bosonic (or fermionic) mode.
The transformation effected by $U$ can equivalently be characterized in state space by a unitary operator $\hat{U}$ that acts on an $N$-mode bosonic (fermionic)  state $|\psi\rangle$ and is defined via the relation
{
\begin{equation}\label{eq:Uc}
    \hat{U}^{\dagger} \, \hat{c}_s \, \hat{U} = (U \boldsymbol{\hat{c}})_s, \quad \forall \, s=1,\cdots,N,
\end{equation}
that is, $\hat{U}^{\dagger} \, \boldsymbol{\hat{c}} \, \hat{U}=U\boldsymbol{\hat{c}}$ in a concise vectorial form. Note the slight abuse of notation in Eq.~\eqref{eq:Uc}, as $\hat{U}$ acts on the tensor product of $N$ operators, meaning we have written, e.g., $\hat{c}_1$ instead of $\hat{c}_1 \otimes \mathds{1} \otimes \cdots \otimes \mathds{1}$ on the left-hand side. The relation in Eq.~\eqref{eq:Uc}} models any linear coupling that conserves the total particle number, namely $\langle \psi | \hat{U}^{\dagger} \boldsymbol{\hat{c}}^\dagger \boldsymbol{\hat{c}} \, \hat{U} | \psi \rangle  =  \langle \psi | \boldsymbol{\hat{c}}^\dagger U^\dagger U \boldsymbol{\hat{c}}| \psi \rangle $ $ =  \langle \psi | \boldsymbol{\hat{c}}^\dagger \boldsymbol{\hat{c}} | \psi \rangle $. Note that the Hilbert space upon which the operator $\hat{U}$ acts depends on whether we deal with bosons or fermions. {As a result, $\hat{U}$ lives in the space of linear operators on the tensor product of $N$ infinite-dimensional Hilbert spaces in the former case, while it lives in the space of linear operators on the tensor product of $N$ Hilbert spaces of dimension $2$ in the later case. We use the same notation $\hat{U}$ in both cases, as the nature of $\hat{U}$ will be clear from the context.
In contrast, the matrix $U$ is the same in both cases, even though the physical implementation of the bosonic and fermionic interferometers defined by the same $U$ can be quite different.}

%%%%%%% Bosonic transition probabilities
The focus of this work concerns  the so-called transition probabilities in an arbitrary linear interferometer. For bosons, these are defined as $B_{\boldsymbol{n}}^{\left( \boldsymbol{i} \right)} \coloneqq | \bra{\boldsymbol{n}} \hat{U} \ket{\boldsymbol{i}} |^2$,
where $\boldsymbol{i} \coloneqq \left(i_1, \cdots, i_N\right) \in \mathbb{N}^N$ and $ \boldsymbol{n} \coloneqq \left(n_1, \cdots, n_N\right) \in \mathbb{N}^N$ denote the vectors of input and output occupation numbers, respectively.
These quantities can equivalently be expressed in phase space as~\cite{scheel2004}
\begin{equation} \label{eq:defBos}
    B_{\boldsymbol{n}}^{\left( \boldsymbol{i} \right)} = \frac{| \per (U_{\boldsymbol{n},\boldsymbol{i}}) |^2 } {\boldsymbol{n}! \, \boldsymbol{i}!},
\end{equation}
where $\boldsymbol{n}! = \prod_s n_s!$, $\boldsymbol{i}!=\prod_s i_s!$, with the product being over $s \in [N]$ and $\per(A)$ stands for the permanent of matrix $A$.
The matrix $U_{\boldsymbol{n},\boldsymbol{i}}$ is defined as follows: we discard all rows of $U$ corresponding to modes $s$ with occupation number $n_s=0$ in vector $\boldsymbol{n}$, before repeating each other row $s$ a number of times corresponding to its associated occupation number $n_s$ in $\boldsymbol{n}$. We do the same for the columns of $U$ but based on vector $\boldsymbol{i}$.
%%%%%%% Fermionic transition probabilities
Similarly, the fermionic transition probabilities $F_{\boldsymbol{n}}^{\left( \boldsymbol{i} \right)} \coloneqq | \bra{\boldsymbol{n}} \hat{U} \ket{\boldsymbol{i}} |^2$ can be expressed as~\cite{Oszmaniec2022}
\begin{equation} \label{eq:defFer}
    F_{\boldsymbol{n}}^{\left( \boldsymbol{i} \right)} = \frac{| \det (U_{\boldsymbol{n},\boldsymbol{i}}) |^2 } {\boldsymbol{n}! \, \boldsymbol{i}!},
\end{equation}
where $\det(A)$ stands for the determinant of matrix $A$.
In Eq.~\eqref{eq:defFer}, the denominator is actually always equal to $1$ since all fermionic occupation numbers are $0$ or $1$.
%%%%%%% Classical transition probabilities
For our purposes, it is worthwhile to look at a third case, mainly, that of classical (i.e., fully distinguishable) particles. The corresponding transition probabilities are given by
\begin{equation} \label{eq:defCla}
    C_{\boldsymbol{n}}^{\left( \boldsymbol{i} \right)} \coloneqq \frac{\per (M_{\boldsymbol{n},\boldsymbol{i}}) } {\boldsymbol{n}! \, \boldsymbol{i}!},
\end{equation}
where $M$ is a matrix whose entries are $M_{kl}=|u_{kl}|^2$ for $k,l=1,\cdots,N$, while the matrix $M_{\boldsymbol{n},\boldsymbol{i}}$ is constructed from $M$ in the same way as $U_{\boldsymbol{n},\boldsymbol{i}}$ in Eqs.~\eqref{eq:defBos} and~\eqref{eq:defFer}.

%%%%%%%%%%%%%%%%%%%%%%%%
%%%%%%%%%%%%%%%%%%%%%%%%
%%%%%%%% RESULTS %%%%%%%
%%%%%%%%%%%%%%%%%%%%%%%%
%%%%%%%%%%%%%%%%%%%%%%%%

\section{Main results}

Our first result is a fundamental complementarity relation that combines bosonic and fermionic interferences.

%%%%% THEOREM 1 %%%%%
\begin{theo}[\textbf{Boson--fermion complementarity}] \label{theo:RecNmain}
    Let $\boldsymbol{i}, \boldsymbol{n} \in \mathbb{N}^N$ be vectors of input and output occupation numbers and $B_{\boldsymbol{n}}^{\left( \boldsymbol{i} \right)}$ (or $F_{\boldsymbol{n}}^{\left( \boldsymbol{i} \right)}$) denote the bosonic (or fermionic) transition probability from $\ket{\boldsymbol{i}}$ to $\ket{\boldsymbol{n}}$ via a given linear interferometer. If $\boldsymbol{i} = \boldsymbol{n} = \boldsymbol{0}$, $B_{\boldsymbol{n}}^{\left( \boldsymbol{i}\right)} = F_{\boldsymbol{n}}^{\left( \boldsymbol{i}\right)} = 1$; else,
    \begin{equation} \label{eq:theo:RecNmain}
        \sum_{\substack{\boldsymbol{j},\boldsymbol{k}  \in \mathbb{N}^N }}  (\!-\!1)^{|\boldsymbol{j}|} \, F_{\boldsymbol{k}}^{\left(\boldsymbol{j} \right)} \, B_{\boldsymbol{n}-\boldsymbol{k}}^{\left( \boldsymbol{i}-\boldsymbol{j} \right)} = 0,
    \end{equation}
    where $|\boldsymbol{j}| \coloneqq \sum_{s=1}^N j_s$.
\end{theo}
\noindent  
As illustrated in Fig.~\ref{fig:Rec}, Theorem~\ref{theo:RecNmain} implies a fundamental link between the bosonic transition probabilities 
$B_{\boldsymbol{n}}^{\left( \boldsymbol{i} \right)}$ and fermionic transition probabilities
$F_{\boldsymbol{n}}^{\left( \boldsymbol{i} \right)}$, which prevails in any linear interferometer $U$. It solely reflects the role of particle statistics and is oblivious to the exact coupling between the particles. 
Interestingly, the relation can be understood in terms of a multivariate discrete convolution between the two sequences $B_{\boldsymbol{n}}^{\left( \boldsymbol{i} \right)}$ and $(\!-\!1)^{|\boldsymbol{i}|} F_{\boldsymbol{n}}^{\left( \boldsymbol{i} \right)}$, which must vanish as a consequence of the bosonic--fermionic duality. Each term in Eq.~\eqref{eq:theo:RecNmain} thus corresponds to converting ${|\boldsymbol{j}|}$ bosons into ${|\boldsymbol{j}|}$  fermions in the interferometer [note also that Eq.~\eqref{eq:theo:RecNmain} is invariant in the swapping of bosons into fermions and vice versa]. Since fermionic occupation numbers cannot exceed one, the sum over $\boldsymbol{j}$ in Eq.~\eqref{eq:theo:RecNmain} is actually restricted to $\boldsymbol{j} \leq \boldsymbol{1}_{[N]}$ (i.e., $j_s \leq 1$, $\forall s \in [N]$), and similarly for $\boldsymbol{k}$.
Here, $\boldsymbol{1}_{[N]}$ denotes a $N$-dimensional vector full of ones. Note that the sum over $\boldsymbol{j}$ is also restricted to $\boldsymbol{j} \leq \boldsymbol{i}$ (i.e., $j_s \leq i_s$, $\forall s \in [N]$) since bosonic occupation numbers cannot be negative. Similarly, the sum over $\boldsymbol{k}$ is restricted to  $\boldsymbol{k}\le \boldsymbol{n}$. Furthermore, the only nonvanishing terms are such that $|\boldsymbol{i}|= |\boldsymbol{n}|$ and $|\boldsymbol{j}|= |\boldsymbol{k}|$ since particle numbers are conserved in the linear interferometer.

Before interpreting Eq.~\eqref{eq:theo:RecNmain} in physical terms, let us introduce a mathematical result inspired by it. If expressed explicitly in terms of $U$ by making use of Eq.~\eqref{eq:defBos} and Eq.~\eqref{eq:defFer}, Eq.~\eqref{eq:theo:RecNmain} reads
\begin{equation} \label{eq:RecNmain3}
\sum_{\boldsymbol{j}, \boldsymbol{k} \in \mathbb{N}^N} (\!-\!1)^{|\boldsymbol{j}|} \, \frac{| \det (U_{\boldsymbol{k},\boldsymbol{j}}) |^2 \, | \per (U_{\boldsymbol{n}-\boldsymbol{k},\boldsymbol{i}-\boldsymbol{j}}) |^2  } {\boldsymbol{k}! \, \boldsymbol{j}! \, (\boldsymbol{n}-\boldsymbol{k})! \, (\boldsymbol{i}-\boldsymbol{j})! } = 0 .
\end{equation}
In the above, we have taken the convention that the determinant or permanent of a non-squared matrix is zero. Furthermore, since the matrix $U_{\boldsymbol{n},\boldsymbol{i}}$ is undefined when $n_s < 0 $ or $i_s < 0$ for some $s \in [N]$, we take the convention that $\per(U_{\boldsymbol{n},\boldsymbol{i}}) = 0$ in such a case.
It is tempting to ask whether Eq.~\eqref{eq:RecNmain3} holds for any matrix $A \in \mathbb{C}^{N \times N}$, rather than unitary matrices. This turns out to be the case, 
and constitutes our second result.
%the following generalization of Theorem~\ref{theo:RecNmain}
%%%%% THEOREM 2 %%%%%
\begin{theo}[\textbf{Permanent--determinant identity}] \label{theo:recNGen}
    Let $\boldsymbol{i}, \boldsymbol{n} \in \mathbb{N}^N$. For any matrix $A \in \mathbb{C}^{N \times N}$,
    \begin{equation} \label{eq:theo:RecNGen}
        \sum_{\boldsymbol{j}, \boldsymbol{k} \in \mathbb{N}^N} (\!-\!1)^{|\boldsymbol{j}|} \frac{| \det (A_{\boldsymbol{k},\boldsymbol{j}}) |^2 \, | \per (A_{\boldsymbol{n}-\boldsymbol{k},\boldsymbol{i}-\boldsymbol{j}}) |^2  } {\boldsymbol{k}! \, \boldsymbol{j}!\, (\boldsymbol{n}-\boldsymbol{k})! \, (\boldsymbol{i}-\boldsymbol{j})! } = 0.
    \end{equation}
\end{theo}

\begin{figure}[t]
\center
    \includegraphics[width=.6\columnwidth]{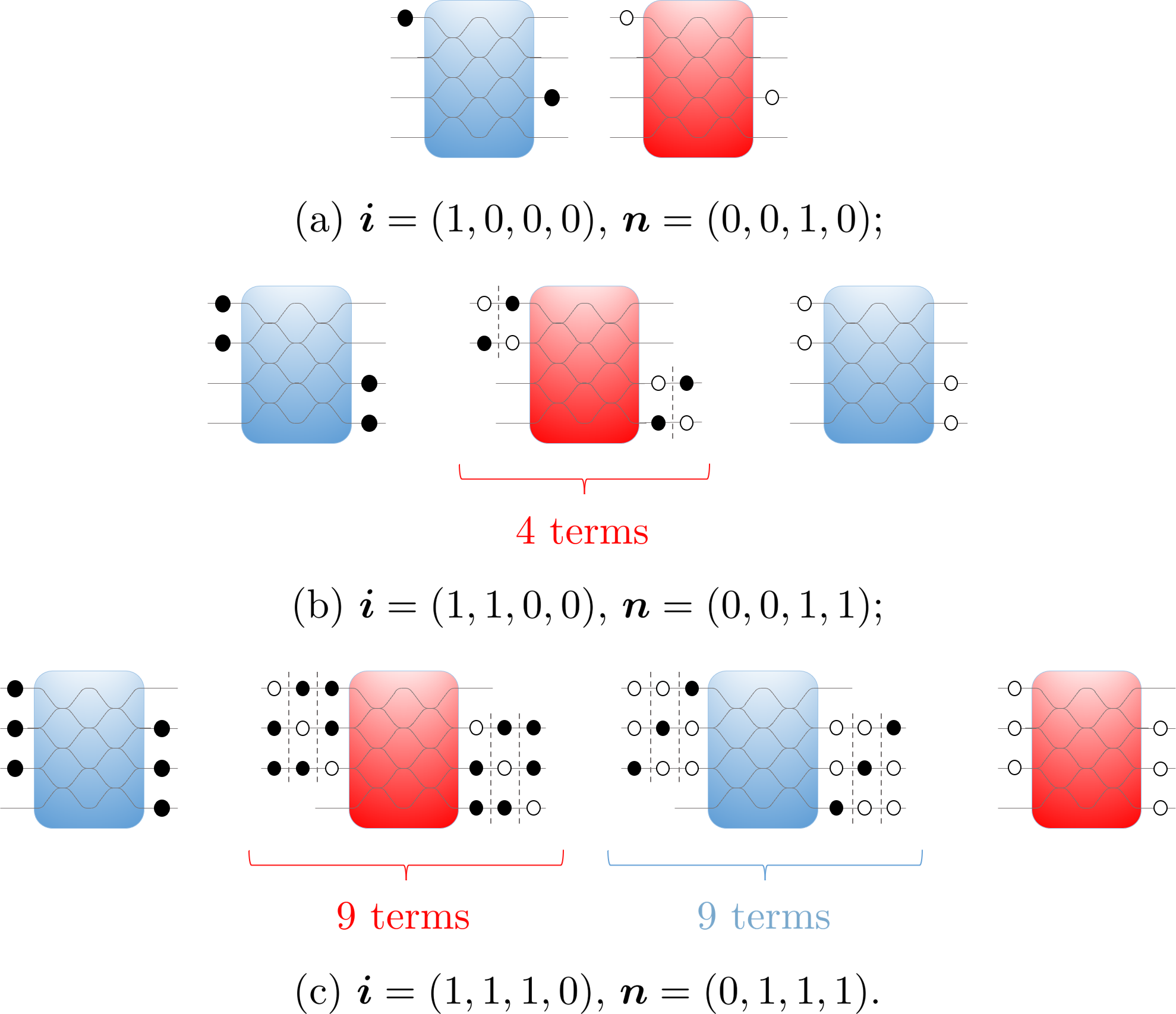}
    %\captionsetup{width=.9\columnwidth}
    \caption{\label{fig:Rec} Examples of the  components of the fundamental relation~\eqref{eq:theo:RecNmain} for an arbitrary interferometer with four modes ($N=4$) and at most one particle per mode. We consider different numbers of particles: (a) $|\boldsymbol{i}| = |\boldsymbol{n}| = 1$, (b) $|\boldsymbol{i}| = |\boldsymbol{n}| = 2$ and (c) $|\boldsymbol{i}| = |\boldsymbol{n}| = 3$. The full (open) circles represent bosons (fermions). The blue and red components correspond, respectively, to the terms with positive and negative signs in Eq.~\eqref{eq:theo:RecNmain}. The dashed vertical lines at the inputs and outputs of each interferometer separate different input and output patterns  (rather than multiple particles per input or output mode); the summation over all input and output patterns is implicit.}
\end{figure}

Equation~\eqref{eq:theo:RecNGen} calls to mind a recurrence relation between permanents and determinants established at the end of the 19\textsuperscript{th} century by Muir~\cite{Muir1897} [see also Ref.~\onlinecite{Chu1989}]. For some fixed value of $N$, denote by $\mathcal{R}_m$ the set of all ordered subsets of $[N]\coloneqq \left\lbrace 1, 2, \cdots ,N\right\rbrace$ with cardinality $m$. Denote by $[A]_{\alpha,\beta}$ the minor of the matrix obtained from $A$ by keeping the rows (columns) whose indices belong to subset $\alpha$ ($\beta$); furthermore, write $[A]_{\alpha}$ for the principal minor $[A]_{\alpha,\alpha}$.
Similarly, denote by $\{A\}_{\alpha,\beta}$ the permanent of the submatrix of $A$ obtained by keeping the rows (columns) whose indices belong to the subset $\alpha$ ($\beta$), and write $\{A\}_{\alpha}$ for $\{A\}_{\alpha,\alpha}$.
In these notations, Muir's relation~\cite{Muir1897} can be recast as follows. For any matrix $A \in \mathbb{C}^{N \times N}$,
\begin{equation} \label{eq:Muir}
    \sum_{m=0}^N (\!-\!1)^m \sum_{\alpha \in \mathcal{R}_m} [A]_{\alpha} \, \{A\}_{\alpha^{\mathrm{c}}} = 0,
\end{equation}
where $\alpha^{\mathrm{c}} \! \coloneqq \! [N] \! \setminus \! \alpha$ is the complementary set of $\alpha$ in $[N]$. Now, by choosing $\boldsymbol{i}=\boldsymbol{n}=\boldsymbol{1}_{[N]}$ in Eq.~\eqref{eq:theo:RecNGen} and using this same notation, one obtains the following corollary of Theorem~\ref{theo:recNGen}.

%%%%% Corollary 3 %%%%%
\begin{cor} \label{cor:MuirSq}
    For any matrix $A \in \mathbb{C}^{N \times N}$,
    \begin{equation} \label{eq:theo:RecNGen2} 
    \sum_{m=0}^N (\!-\!1)^m \sum_{\alpha, \beta \in \mathcal{R}_m} \big| [A]_{\alpha,\beta}\big|^2 \; \big| \{A\}_{\alpha^{\mathrm{c}},\beta^{\mathrm{c}}} \big|^2 = 0.
    \end{equation}
\end{cor}
\noindent While there are clear similarities between Eq.~\eqref{eq:Muir} and Eq.~\eqref{eq:theo:RecNGen2}, the latter does not seem to be a trivial consequence of the former (and vice versa). In particular, Eq.~\eqref{eq:theo:RecNGen2} involves all the minors of $A$, while Eq.~\eqref{eq:Muir} includes only its principal minors.

%%%%%%%%%%%%%%%%%%%%%%%%
%%%%%%%%%%%%%%%%%%%%%%%%
%%%% INTERPRETATION %%%%
%%%%%%%%%%%%%%%%%%%%%%%%
%%%%%%%%%%%%%%%%%%%%%%%%

\section{Physical interpretation}

\subsection{Comparison with classical interferences}

It is instructive to interpret the different terms of Eq.~\eqref{eq:theo:RecNmain} by comparing the latter with its counterpart for classical particles. In that case,
the laws of probability simply assert that
\begin{equation} \label{eq:theo:RecNclas}
    C_{\boldsymbol{n}}^{\left( \boldsymbol{i} \right)} = \sum_{\substack{\boldsymbol{k} \in \mathbb{N}^N }} C_{\boldsymbol{k}}^{\left(\boldsymbol{j} \right)} \, C_{\boldsymbol{n}-\boldsymbol{k}}^{\left( \boldsymbol{i}-\boldsymbol{j} \right)}, \quad \forall \, \boldsymbol{j} \in \mathbb{N}^N, \quad \boldsymbol{j}\le \boldsymbol{i}.
\end{equation}
That is, $C_{\boldsymbol{n}}^{\left( \boldsymbol{i} \right)}$ can be computed by breaking down the input pattern $\boldsymbol{i}$ into any two patterns $\boldsymbol{j}$ and $\boldsymbol{i-j}$  before summing over all the possible cases leading to pattern $\boldsymbol{n}$ at the output, which gives the convolution of Eq.~\eqref{eq:theo:RecNclas} [see, for instance, Section~D of Ref.~\onlinecite{Jabbour-Cerf-PRR-2021} for details in the case $N=2$]. We now understand that each term in the summation over $\boldsymbol{j}$ in Eq.~\eqref{eq:theo:RecNmain} is analogous to the right-hand side of Eq.~\eqref{eq:theo:RecNclas}, albeit without any negative sign. Indeed, for each $\boldsymbol{j}$ in Eq.~\eqref{eq:theo:RecNmain}, we have a sum over all output patterns of  $\boldsymbol{n}$ particles composed of $\boldsymbol{k}$ fermions and $\boldsymbol{n-k}$ bosons.
It must be emphasized that interferences are accounted for by assigning negative signs to some probabilities in Eq.~\eqref{eq:theo:RecNmain}, which is in sharp contrast with the usual modeling of quantum interferences in terms of complex probability amplitudes [see Ref.~\onlinecite{Jabbour-Cerf-PRR-2021}]. The existence of {probabilities with negative prefactors} in Eq.~\eqref{eq:theo:RecNmain}
further indicates a profound departure from classical interferences as modeled in Eq.~\eqref{eq:theo:RecNclas}, in addition to the fact that Eq.~\eqref{eq:theo:RecNmain} must necessarily incorporate both bosons and fermions.

Let us now analyze the implications of Theorem~\ref{theo:RecNmain}
in the special case where there is at most one particle per mode, which is when the boson--fermion complementarity is the most evident since there cannot anyway be more than one fermion per mode.
Consider first the simplest case of a single particle entering an arbitrary $N$-mode linear interferometer. With no loss of generality, we may consider the probability of detecting the particle in the first output mode when it is sent in the first input mode. Eq.~\eqref{eq:theo:RecNmain} simply yields  $B_1^{(1)}-F_1^{(1)}=0$,
where the subscript (superscript) denotes here the output (input) mode assignment of the single particle; see Fig.~\ref{fig:Rec}a where we illustrate  $B_3^{(1)}$ and $F_3^{(1)}$. As expected, the particle statistics is irrelevant for a single particle, so that
\begin{equation} \label{eq:1part}
B_1^{(1)} = F_1^{(1)}  = C_1^{(1)},
\end{equation}
where $C_1^{(1)} = |u_{11}|^2$ denotes the classical $1\to 1$ transition probability in the $N$-mode interferometer $U$.

\subsection{Two-particle interferences}
Things become nontrivial for two particles.
We consider with no loss of generality the probability of detecting one particle in each of the first two output modes given that a particle is sent in each of the first two input modes, that is $B/F/C^{(12)}_{12}$.
Eq.~\eqref{eq:theo:RecNmain} yields (see Fig.~\ref{fig:Rec}b displaying, for instance, $B^{(12)}_{34}$ and $F^{(12)}_{34}$)
\begin{equation}
    B^{(12)}_{12} - B^{(2)}_2 F^{(1)}_1 - B^{(2)}_1 F^{(1)}_2 - B^{(1)}_2 F^{(2)}_1 - B^{(1)}_1 F^{(2)}_2 
+F^{(12)}_{12} = 0.
\end{equation}
Using $B^{(i)}_{j}=F^{(i)}_{j}=C^{(i)}_{j}=|u_{ji}|^2$ as in Eq.~\eqref{eq:1part}, we obtain
\begin{align}
B^{(12)}_{12} +F^{(12)}_{12} 
&= 2 \,|u_{11}|^2|u_{22}|^2 + 2 \,|u_{12}|^2|u_{21}|^2  \nonumber\\
&= 2 \, \text{perm}(M)
\nonumber\\
&= 2 \, C^{(12)}_{12}.
\label{eq_N=2_complementarity}
\end{align}
This fundamental complementarity relation implies that the classical 2-particle transition probability is the \textit{average} between the bosonic and fermionic probabilities. In other words, the tendency of bosons to bunch in modes 1 and 2 (measured by $C^{(12)}_{12} -B^{(12)}_{12}$, i.e., the defect of coincidence probability with respect to classical particles)
is precisely equal to the tendency of fermions to antibunch in modes 1 and 2 (measured by $F^{(12)}_{12} - C^{(12)}_{12}$, i.e., the excess of coincidence probability with respect to classical particles). In the Hong-Ou-Mandel effect ($N=2$), we have indeed 
$B^{(12)}_{12} = 0$, $F^{(12)}_{12} =1$, and $C^{(12)}_{12}=1/2$, but remember that Eq.~\eqref{eq_N=2_complementarity} holds true for any two input and output modes of any linear interferometer (with arbitrary $N$).
Note that there exist situations (for $N>2$) where two bosons antibunch, i.e., $ B^{(12)}_{12}>C^{(12)}_{12}$, but then Eq.~\eqref{eq_N=2_complementarity} implies that two fermions must bunch for the same interferometer, i.e., 
$F^{(12)}_{12} < C^{(12)}_{12}$. We name the two-particle instances where 
$B^{(12)}_{12}<C^{(12)}_{12}<F^{(12)}_{12}$ as
\textit{natural}, whereas those with 
$B^{(12)}_{12}>C^{(12)}_{12}>F^{(12)}_{12}$ are called \textit{antinatural}.

{

\subsection{Three-particle interferences}

In the case of three particles, we may consider for example the probability of detecting one particle in each of the first three output modes when a particle is sent in each of the first
three input modes, that is $B/F/C^{(123)}_{123}$. Equation~\eqref{eq:theo:RecNmain} then gives (see Fig.~\ref{fig:Rec}c displaying, for instance, $B^{(123)}_{234}$ and $F^{(123)}_{234}$)
\begin{equation} \label{eq_N=3_complementarity}
B^{(123)}_{123} - \sum_{i,j=1}^3 B^{(\neg i)}_{\neg j} F^{(i)}_{j} + \sum_{i,j=1}^3 B^{(i)}_{j} F^{(\neg i)}_{\neg j} 
-F^{(123)}_{123} = 0,
\end{equation}
where $B^{(\neg i)}_{\neg j}$
(resp. $F^{(\neg i)}_{\neg j}$) stands for the two-particle bosonic (resp. fermionic) transition probability when the particles are sent in the two modes $\neg i \coloneqq [3]\!\setminus\! i$  and are detected in the two modes $\neg j \coloneqq [3]\!\setminus\! j$. It is not surprising that the formula is less straightforward for three (or more) particles since bosonic or fermionic statistics intrinsically concerns the exchange of two particles. Nevertheless, we may rewrite Eq.~\eqref{eq_N=3_complementarity} as
\begin{equation} 
B^{(123)}_{123} 
-F^{(123)}_{123} = \sum_{i,j=1}^3 C^{(i)}_{j} \left( B^{(\neg i)}_{\neg j}  -  F^{(\neg i)}_{\neg j} \right)  \label{eq:diff_for_N=3_a},
\end{equation}
where we have used the fact that single-particle transition probabilities are the same for bosons, fermions, and classical particles, see Eq.~\eqref{eq:1part}. Equation~\eqref{eq:diff_for_N=3_a} implies that the ``$B-F$'' difference  for three particles is proportional to the average of all possible two-particle ``$B-F$'' differences weighted by the associated probabilities $C^{(i)}_{j}/3$.
Hence, if the two-particle transitions are, on average, natural (i.e., $B<C<F$), then the three-particle transitions must be natural too (i.e., $B<F$; here, $C$ is not the average of $B$ and $F$). Of course, the same is true for antinatural transitions.
This can be illustrated with three particles impinging on a 3-mode Fourier interferometer,  $u_{kl} = e^{- 2 \pi i k l/3}/\sqrt{3}$, where $i$ represents the imaginary unit. The three-particle transition probabilities are $B^{(123)}_{123}=1/3$ and $F^{(123)}_{123}=1$, so the three-particle behavior is natural. This is consistent with the fact that all two-particle transition probabilities are natural too (in this simple example, they are all equal to  
$B^{(12)}_{12}=1/9$ and
$F^{(12)}_{12}=1/3$).

We may also reexpress the right-hand side of Eq.~\eqref{eq:diff_for_N=3_a} without any two-particle fermionic transition probabilities, that is, 
\begin{eqnarray}
B^{(123)}_{123} 
-F^{(123)}_{123} &=& 2 \sum_{i,j=1}^3 C^{(i)}_{j} \left( B^{(\neg i)}_{\neg j}  -  C^{(\neg i)}_{\neg j} \right)  \label{eq:diff_for_N=3_b} \\
&=& 2 \sum_{i,j=1}^3 C^{(i)}_{j}  B^{(\neg i)}_{\neg j}  - 6 \, C^{(123)}_{123}
\label{eq:diff_for_N=3_c}
\end{eqnarray}
where we have used the relation between $B^{(\neg i)}_{\neg j}$ and $F^{(\neg i)}_{\neg j}$ for two particles coming from Eq.~\eqref{eq_N=2_complementarity}, as well as the Laplace expansion formula for the permanent, namely
$C^{(123)}_{123}=\sum_{i=1}^3 C^{(i)}_{j}
C^{(\neg i)}_{\neg j}$,  $\forall j\in [3]$. Conversely,
we may write relations similar to Eqs.~\eqref{eq:diff_for_N=3_b} or~\eqref{eq:diff_for_N=3_c} but in terms of $F^{(\neg i)}_{\neg j}$ instead of $B^{(\neg i)}_{\neg j}$, namely
\begin{eqnarray} %\nonumber
B^{(123)}_{123} 
-F^{(123)}_{123} 
&=& 2 \sum_{i,j=1}^3 C^{(i)}_{j} \left( C^{(\neg i)}_{\neg j}  -  F^{(\neg i)}_{\neg j} \right)  \label{eq:second_diff_for_N=3_b} \\
&=& 6 \, C^{(123)}_{123} - 2 \sum_{i,j=1}^3  C^{(i)}_{j}  F^{(\neg i)}_{\neg j}  
\label{eq:second_diff_for_N=3_c}
\end{eqnarray}
This last expression can even be further simplified in case the size of the interferometer is $N=3$ (the three particles are sent in a 3-mode interferometer). Then, we have $F^{(\neg i)}_{\neg j}=|u_{ji}|^2=C^{(i)}_{j}$  as a consequence of the unitarity of $U$ (more generally, the state of $N-1$ fermions in $N$ modes is equivalent to the state of a single fermion -- a hole in the sea of $N$ fermions -- which itself behaves as a classical particle in the interferometer). Therefore, we may write the three-particle ``$B-F$'' difference in terms of classical transition probabilities only, namely
\begin{equation} 
B^{(123)}_{123} 
-F^{(123)}_{123} 
= 6 \, C^{(123)}_{123} - 2 \sum_{i,j=1}^3  \left(C^{(i)}_{j}\right)^2    
\label{eq:third_diff_for_N=3_c}
\end{equation}
This relation can easily be verified in the above example of three particles impinging on a 3-mode Fourier interferometer, as we have
$C^{(123)}_{123}=2/9$ and all nine single-particle transition probabilities $C^{(i)}_j=1/3$.

The generalization to four (or more) particles seems increasingly involved, but, from these examples, we observe a clear structure in the outcome of Eq.~\eqref{eq:theo:RecNmain} as it constraints the sum ``$B+F$'' of the bosonic and fermionic probabilities for even particle numbers, whereas it constraints their difference ``$B-F$'' for odd particle numbers.
This hints at the existence of recurrence relations for these sums and differences, which deeply connect bosonic and fermionic interferences.
}

%%%%%%%%%%%%%%%%%%%%%%%%
%%%%%%%%%%%%%%%%%%%%%%%%
%%%%%%%% PROOFS %%%%%%%%
%%%%%%%%%%%%%%%%%%%%%%%%
%%%%%%%%%%%%%%%%%%%%%%%%

\section{Proofs}

\subsection{Proof of Theorem~\ref{theo:RecNmain}}

Theorem~\ref{theo:RecNmain} expresses a relation that combines bosonic and fermionic multiparticle interferences.
%%%%%%% Generating function
The corresponding transition probabilities are, in general, complicated to express analytically. For instance, with bosons, the difficulty is that the input and output Fock states are non-Gaussian, while the unitary $\hat{U}$ modeling the linear interferometer is Gaussian. As was shown in Ref.~\onlinecite{Jabbour-Cerf-PRR-2021} for the special case of two modes ($N=2$), this problem can be circumvented by exploiting the generating function (GF) of the transition probabilities.
The $2 N$-variate GF of $B_{\boldsymbol{n}}^{\left( \boldsymbol{i} \right)}$ is defined as
\begin{equation}
    g(\boldsymbol{x},\boldsymbol{z}) \coloneqq \sum_{\boldsymbol{i} \in \mathbb{N}^N} \sum_{\boldsymbol{n} \in \mathbb{N}^N}  \left( \prod_{s=1}^N x_s^{i_s} \right) \left( \prod_{r=1}^N z_r^{n_r} \right) B_{\boldsymbol{n}}^{\left( \boldsymbol{i} \right)},
    \label{eq:def:g}
\end{equation}
where $\boldsymbol{x}, \boldsymbol{z} \in [0,1)^N$.
This function of the  vectors of input and output dual variables $\boldsymbol{x} \coloneqq \left(x_1, \cdots, x_N\right)$ and $ \boldsymbol{z} \coloneqq \left(z_1, \cdots, z_N\right)$ encapsulates all transition probabilities since we have
\begin{equation}
\partial_{x_1}^{\,i_1}\cdots \partial_{x_N}^{\,i_N}\, \partial_{z_1}^{\,n_1}\cdots \partial_{z_N}^{\,n_N} \, g(\boldsymbol{x},\boldsymbol{z})\big|_{\boldsymbol{x}=\boldsymbol{z}=\boldsymbol{0}} = \boldsymbol{i}! \, \boldsymbol{n}! \, B_{\boldsymbol{n}}^{\left( \boldsymbol{i} \right)},
\end{equation}
where $\boldsymbol{0} \coloneqq \left(0, \cdots, 0\right)$. The power of the GF is witnessed by the fact that it becomes proportional to a transition probability between input and output states that are now Gaussian, making it easy to compute (see Ref.~\onlinecite{Jabbour-Cerf-PRR-2021} for $N=2$), which in turn gives rise to an elegant relation between the transition probabilities as materialized by Theorem~\ref{theo:RecNmain}.

To prove Theorem~\ref{theo:RecNmain}, we begin by expressing the GF of the bosonic transition probabilities. This is encompassed in the following Lemma.

%%%%% LEMMA 1 %%%%%
\begin{lem} \label{lem:GenFunc}
    Let $\boldsymbol{x}, \boldsymbol{z} \in [0,1)^N$. The generating function $g(\boldsymbol{x},\boldsymbol{z})$ of the bosonic transition probabilities via a linear interferometer given by the unitary matrix $U$ satisfies
    \begin{equation}\label{eq:lem:GenFunc}
        \frac{1}{g(\boldsymbol{x},\boldsymbol{z})} = \sum_{m=0}^N (\!-\!1)^m \!\! \sum_{\alpha, \beta \in \mathcal{R}_m} [Z]_{\beta} \, \Big| [U]_{\beta,\alpha}\Big|^2 \,  [X]_{\alpha},
    \end{equation}
    where $X = \mathrm{diag}(\boldsymbol{x})$, $Z = \mathrm{diag}(\boldsymbol{z})$, and $[U]_{\beta,\alpha}$ is the minor of the matrix obtained from $U$ by keeping the rows (columns) whose indices belong to the set $\beta$ ($\alpha$).
\end{lem}
\noindent Remember that $\mathcal{R}_m$ stands for the set of all ordered subsets of $[N]\coloneqq \left\lbrace 1, 2, \cdots ,N\right\rbrace$ with cardinality $m$. There are $2^N$ possible subsets of $[N]$, starting from the empty set $\emptyset$ to the entire set $[N]$, and we group them in Eq.~\eqref{eq:lem:GenFunc} according to their cardinality. For example, $\alpha\in \mathcal{R}_m$ means that $\alpha \subseteq [N]$ and $|\alpha|=m$. The size of $\mathcal{R}_m$ is the number of bipartitions of $[N]$ into $m$ elements and $N-m$ elements, that is, $\binom{N}{m}$.
\begin{proof}[Proof of Lemma~\ref{lem:GenFunc}]
    First, a generalization of MacMahon's master theorem~\cite{MacMahon,Chabaud2022} implies that Eq.~\eqref{eq:def:g} can be rewritten as
    \begin{equation}
        g(\boldsymbol{x},\boldsymbol{z}) = (\det \left[\mathds{1}_N - U^{\dagger} Z U X \right])^{-1},    
    \end{equation}
    where $\mathds{1}_N$ is the identity matrix of dimension $N$.
    Defining $\bar{g}(\boldsymbol{x},\boldsymbol{z}) \coloneqq \left( g(\boldsymbol{x},\boldsymbol{z}) \right)^{-1}$ for notational simplicity, and using the Leibniz formula for the determinant, we have
    \begin{equation}
        \bar{g}(\boldsymbol{x},\boldsymbol{z}) = \sum_{\sigma \in S_N} \left( \mathrm{sgn}(\sigma) \prod_{i=1}^N \left[ \delta_{i,\sigma_i} - c_{i,\sigma_i} \right] \right),
    \end{equation}
    where $S_N$ stands for the symmetric group and 
    \begin{equation}
        c_{i,\sigma_i} \coloneqq \sum_{k=1}^N {u_{ki}^*} u_{k \sigma_i} x_{\sigma_i} z_k .
    \end{equation}
    Denote by $\mathcal{P}_N$ the power set of $[N]$ and define $\bar{\nu} \coloneqq [N] \setminus \nu$, where $\nu\in \mathcal{P}_N$ denote any subset of $[N]$ and $\bar{\nu}\in \mathcal{P}_N$ its complementary subset. We have
    \begin{equation}
        \prod_{i=1}^N \left[ \delta_{i,\sigma_i} - c_{i,\sigma_i} \right] = \sum_{\nu \in \mathcal{P}_N} (-1)^{|\bar{\nu}|} \bigg( \prod_{i \in \nu} \delta_{i,\sigma_i} \bigg) \bigg( \prod_{j \in \bar{\nu}} c_{j,\sigma_j} \bigg).
        \label{eq:intermediate1}
    \end{equation}
    The right-hand most product in the above can be written as
    \begin{equation}
        \prod_{j \in \bar{\nu}} c_{j,\sigma_j} = \sum_{\alpha \in [N]^{|\bar{\nu}|}} \prod_{k=1}^{|\bar{\nu}|} {u_{\alpha_k \bar{\nu}_k}^*} u_{\alpha_k \sigma_{\bar{\nu}_k}} x_{\sigma_{\bar{\nu}_k}} z_{\alpha_k},
    \end{equation}
    where $[N]^{|\bar{\nu}|}$ represents the $|\bar{\nu}|$-fold Cartesian product of the set $[N]$ with itself.
    Using this and splitting the summation over $\nu \in \mathcal{P}_N$ in Eq.~\eqref{eq:intermediate1} into two summations, one over the cardinality $m = |\nu|$ going from $0$ to $N$, and one over all subsets $\nu$ of cardinality $m$, we obtain
    \begin{equation}
        \bar{g}(\boldsymbol{x},\boldsymbol{z}) = \sum_{m=0}^N (-1)^{N-m} \sum_{\substack{\nu \in \mathcal{P}_N \\ |\nu| = m}} \sum_{\alpha \in [N]^{N-m}} \sum_{\sigma \in S_N} \mathrm{sgn}(\sigma) \bigg( \prod_{i \in \nu} \delta_{i,\sigma_i} \bigg) \prod_{k=1}^{N-m} {u_{\alpha_k \bar{\nu}_k}^*} u_{\alpha_k \sigma_{\bar{\nu}_k}} x_{\sigma_{\bar{\nu}_k}} z_{\alpha_k}.
    \end{equation}
    The set $[N]^{|\nu|}$ is the set of all possible $|\nu|$-element sets , whose elements are between $1$ and $N$, as well as all their possible permutations. If we denote by $Q^{(N)}_{|\nu|}$ the set of all ordered $|\nu|$-element sets whose elements are between $1$ and $N$, we have
    \begin{equation}
        \sum_{\alpha \in [N]^{N-m}} \cdots = \sum_{\beta \in Q^{(N)}_{N-m}} \sum_{\alpha \in \pi \left(\beta\right)} \cdots,
    \end{equation}
    where $\pi \left(\beta\right)$ represents the set of all possible permutations of $\beta$. With this, we have
    \begin{equation}
        \bar{g}(\boldsymbol{x},\boldsymbol{z}) = \sum_{m=0}^N (-1)^{N-m} \sum_{\substack{\nu \in \mathcal{P}_N \\ |\nu| = m}} \sum_{\beta \in Q^{(N)}_{N-m}} \sum_{\sigma \in S_N} \mathrm{sgn}(\sigma) \bigg( \prod_{i \in \nu} \delta_{i,\sigma_i} \bigg) \sum_{\alpha \in \pi \left(\beta\right)} \prod_{k=1}^{N-m} {u_{\alpha_k \bar{\nu}_k}^*} u_{\alpha_k \sigma_{\bar{\nu}_k}} x_{\sigma_{\bar{\nu}_k}} z_{\alpha_k}.
    \end{equation}
    Now, we may factor out
    \begin{equation}
            \sum_{\alpha \in \pi \left(\beta\right)} \prod_{k=1}^{N-m} {u_{\alpha_k \bar{\nu}_k}^*} u_{\alpha_k \sigma_{\bar{\nu}_k}} x_{\sigma_{\bar{\nu}_k}} z_{\alpha_k}
            = \bigg( \sum_{\alpha \in \pi \left(\beta\right)} \prod_{k=1}^{N-m} {u_{\alpha_k \bar{\nu}_k}^*} u_{\alpha_k \sigma_{\bar{\nu}_k}} \bigg) \bigg( \prod_{j \in \bar{\nu}} x_{\sigma_j} \bigg) [Z]_{\beta},
    \end{equation}
    so that
    \begin{equation}
            \bar{g}(\boldsymbol{x},\boldsymbol{z}) = \sum_{m=0}^N (-1)^{N-m} \sum_{\substack{\nu \in \mathcal{P}_N \\ |\nu| = m}} \sum_{\beta \in Q^{(N)}_{N-m}} [Z]_{\beta} \sum_{\alpha \in \pi \left(\beta\right)} \sum_{\sigma \in S_N} \mathrm{sgn}(\sigma) \bigg( \prod_{i \in \nu} \delta_{i,\sigma_i} \bigg)
            \bigg( \prod_{k=1}^{N-m} {u_{\alpha_k \bar{\nu}_k}^*} u_{\alpha_k \sigma_{\bar{\nu}_k}} \bigg) \bigg( \prod_{j \in \bar{\nu}} x_{\sigma_j} \bigg).
    \end{equation}
    The factor $\prod_{i \in \nu} \delta_{i,\sigma_i}$ eliminates some of the $\sigma$'s in the summation over all permutations of $S_N$, depending on $\nu$. The index $\sigma_j$ will actually span all the permutations $\theta$ of $\bar{\nu}$. It also happens that $\mathrm{sgn}(\sigma)$ which appears in the summation over $\sigma$ will be equal to $\mathrm{sgn}(\theta)$, taking the convention that $\bar{\nu}$ is always ordered. Having this in mind, we get
    \begin{equation}
        \begin{aligned}
            \bar{g}(\boldsymbol{x},\boldsymbol{z}) & = \sum_{m=0}^N (-1)^{N-m} \sum_{\substack{\nu \in \mathcal{P}_N \\ |\nu| = m}} \sum_{\beta \in Q^{(N)}_{N-m}} [Z]_{\beta} \sum_{\alpha \in \pi \left(\beta\right)} \sum_{\theta \in \pi \left(\bar{\nu}\right)}
            \mathrm{sgn}(\theta) \bigg( \prod_{k=1}^{N-m} {u_{\alpha_k \bar{\nu}_k}^*} u_{\alpha_k \theta_k} \bigg) \bigg( \prod_{j=1}^{|\bar{\nu}|} x_{\theta_j} \bigg) \\
            & = \sum_{m=0}^N (-1)^{N-m} \sum_{\beta \in Q^{(N)}_{N-m}} \sum_{\substack{\bar{\nu} \in \mathcal{P}_N \\ |\bar{\nu}| = N-m}} [Z]_{\beta} [X]_{\bar{\nu}} \sum_{\theta \in \pi \left(\bar{\nu}\right)} \mathrm{sgn}(\theta) \sum_{\alpha \in \pi \left(\beta\right)} \bigg( \prod_{k=1}^{N-m} {u_{\alpha_k \bar{\nu}_k}^*} u_{\alpha_k \theta_k} \bigg)
        \end{aligned}
        \label{eq:intermediate3}
    \end{equation}   
    The last product in $\bar{g}(\boldsymbol{x},\boldsymbol{z})$ can be written as
    \begin{equation}
        \prod_{k=1}^{N-m } {u_{\alpha_k \bar{\nu}_k}^*} u_{\alpha_k \theta_k} = \prod_{k=1}^{N-m } {(U \left(\alpha,\bar{\nu}\right)_{k k})^*} U \left(\alpha,\theta\right)_{k k} = \bigg( \prod_{k=1}^{N-m } {(U \left(\alpha,\bar{\nu}\right)_{k k})^*} \bigg) \bigg( \prod_{j=1}^{N-m } U \left(\alpha,\theta\right)_{j j} \bigg),
    \end{equation}
   where we have used the notation $U(\alpha,\beta)$ to denote the submatrix of $U$ obtained by keeping the rows (columns) whose indices belong to the subset $\alpha$ ($\beta$). Hence, using the Leibniz formula for the determinant, 
    \begin{equation}
        \sum_{\theta \in \pi \left(\bar{\nu}\right)} \mathrm{sgn}(\theta) \prod_{k=1}^{N-m} {u_{\alpha_k \bar{\nu}_k}^*} u_{\alpha_k \theta_k} = \bigg( \prod_{k=1}^{N-m} {(U \left(\alpha,\bar{\nu}\right)_{k k})^*} \bigg) [U]_{\alpha,\bar{\nu}}.
        \label{eq:intermediate2}
    \end{equation}
    Since exchanging two rows of a matrix changes the sign of its determinant, we have
    \begin{equation}
        [U]_{\beta,\bar{\nu}} = \mathrm{sgn}(\alpha) [U]_{\alpha,\bar{\nu}}, \quad \alpha \in \pi \left(\beta\right),
    \end{equation}
    so that Eq.~\eqref{eq:intermediate2} can be rewritten as
    \begin{equation}
        \begin{aligned}
             \bigg( \prod_{k=1}^{N-m} {(U \left(\alpha,\bar{\nu}\right)_{k k})^*} \bigg) \mathrm{sgn}(\alpha) [U]_{\beta,\bar{\nu}}, \quad \alpha \in \pi \left(\beta\right).
        \end{aligned}
    \end{equation}
    Then, the summation over $\alpha$ can be simplified by using again the Leibniz formula, which gives
    \begin{equation}
            \sum_{\alpha \in \pi \left(\beta\right)} \sum_{\theta \in \pi \left(\bar{\nu}\right)} \mathrm{sgn}(\theta) \prod_{k=1}^{N-m} {u_{\alpha_k \bar{\nu}_k}^*} u_{\alpha_k \theta_k}
            = \sum_{\alpha \in \pi \left(\beta\right)} \bigg( \prod_{k=1}^{N-m} {(U \left(\alpha,\bar{\nu}\right)_{k k})^*} \bigg) \mathrm{sgn}(\alpha) [U]_{\beta,\bar{\nu}} {= ([U]_{\beta,\bar{\nu}})^* [U]_{\beta,\bar{\nu}} = \left|[U]_{\beta,\bar{\nu}}\right|^2}.
    \end{equation}
    Using this, Eq.~\eqref{eq:intermediate3} reduces to
    \begin{equation}
        \bar{g}(\boldsymbol{x},\boldsymbol{z}) = \sum_{m=0}^N (-1)^{N-m} \sum_{\beta \in Q^{(N)}_{N-m}} \sum_{\substack{\bar{\nu} \in \mathcal{P}_N \\ |\bar{\nu}| = N-m}} [Z]_{\beta} [X]_{\bar{\nu}} {\left|[U]_{\beta,\bar{\nu}}\right|^2},
    \end{equation}
    which is essentially equivalent to
    \begin{equation}
        \bar{g}(\boldsymbol{x},\boldsymbol{z}) = \sum_{m=0}^N (-1)^m \sum_{\alpha,\beta \in \mathcal{R}_m} [Z]_{\beta} {\left| [U]_{\beta,\alpha} \right|^2}  [X]_{\alpha}.
    \end{equation}
   This concludes the proof of Lemma~\ref{lem:GenFunc}.
\end{proof}

We are now in position to prove Theorem~\ref{theo:RecNmain}. Equation~\eqref{eq:lem:GenFunc} of Lemma~\ref{lem:GenFunc} can obviously be rewritten as
\begin{equation} \label{eq:GenFunc2}
    \sum_{m=0}^N (\!-\!1)^m \!\! \sum_{\alpha, \beta \in \mathcal{R}_m} \Big| [U]_{\beta,\alpha}\Big|^2 \, [X]_{\alpha} [Z]_{\beta} \, \, g(\boldsymbol{x},\boldsymbol{z}) = 1.
\end{equation}
Now, using the shifting property of the GFs, it can easily be shown that multiplying $g(\boldsymbol{x},\boldsymbol{z})$ by $x_s$ for $s\in\alpha$ corresponds to decreasing the index $i_s$ of $B_{\boldsymbol{n}}^{\left( \boldsymbol{i} \right)}$ by one unit, while multiplying it by $z_s$ for $s\in\beta$ corresponds to decreasing the index $n_s$ by one unit. In addition, we know that the GF of the product $\delta_{\boldsymbol{i},\boldsymbol{0}} \, \delta_{\boldsymbol{n},\boldsymbol{0}}$ of two multivariate Kronecker deltas is $1$. Therefore, moving back from the vectors of dual variables $\boldsymbol{x}$ and $\boldsymbol{z}$ to the vectors of occupation numbers $\boldsymbol{i}$ and $\boldsymbol{n}$,  Eq.~\eqref{eq:GenFunc2} is mapped into the following relation:
\begin{equation} \label{eq:RecNmain2}
        \sum_{m=0}^N (\!-\!1)^{m} \!\! \sum_{\substack{\alpha, \beta \in \mathcal{R}_m }} \Big| [U]_{\beta,\alpha}\Big|^2 \,B_{\boldsymbol{n}-\boldsymbol{1}_{\beta}}^{\left( \boldsymbol{i}-\boldsymbol{1}_{\alpha} \right)} = \delta_{\boldsymbol{i},\boldsymbol{0}} \, \delta_{\boldsymbol{n},\boldsymbol{0}},
    \end{equation}
where $\boldsymbol{1}_\alpha$ denotes the $N$-dimensional vector with ones at all positions $s \in \alpha$ and zeros elsewhere (and similarly for $\boldsymbol{1}_\beta$).
Note that $B_{\boldsymbol{n}}^{\left( \boldsymbol{i}\right)}$ can be  nonzero only if $|\boldsymbol{i}| = |\boldsymbol{n}|$ (with $|\boldsymbol{i}| \coloneqq \sum_{s} i_s$ and $|\boldsymbol{n}| \coloneqq \sum_{s} n_s$) since $\hat{U}$ preserves the particle number. Thus, the contributing terms in Eq.~\eqref{eq:RecNmain2} have $|\boldsymbol{1}_{\alpha}| = |\boldsymbol{1}_{\beta}|=m$ when $\alpha, \beta \in \mathcal{R}_m$. Furthermore, it is implicitly assumed that all terms in the above summations implying a negative entry in the vectors $\boldsymbol{i}-\boldsymbol{1}_{\alpha}$ or $\boldsymbol{n}-\boldsymbol{1}_{\beta}$ vanish (indeed, occupation numbers cannot be negative). Thus, for a given $\boldsymbol{i}$, we take any subset $\alpha$ in the summation over $\mathcal{R}_m$ only if $i_s > 0$, $\forall s \in \alpha$. Similarly, for a given $\boldsymbol{n}$, we keep a subset $\beta$ only if $n_r > 0$, $\forall r \in \beta$.
\par

While Eq.~\eqref{eq:RecNmain2} seems to involve only bosonic probabilities at first sight, one cannot help but notice that the quantities $| [U]_{\beta,\alpha}|^2$ can in fact be viewed as fermionic transition probabilities.
Indeed, using the latter's definition in Eq.~\eqref{eq:defFer},  Eq.~\eqref{eq:RecNmain2} can be rewritten in a much more symmetric form as
\begin{equation}
    \sum_{\substack{\boldsymbol{j},\boldsymbol{k}  \in \mathbb{N}^N }}  (\!-\!1)^{|\boldsymbol{j}|} \, F_{\boldsymbol{k}}^{\left(\boldsymbol{j} \right)} \, B_{\boldsymbol{n}-\boldsymbol{k}}^{\left( \boldsymbol{i}-\boldsymbol{j} \right)} = 0.
\end{equation}
In the above, we exploited the fact that the fermionic occupation numbers cannot exceed one, hence the sum over $\boldsymbol{j}$ is actually restricted to $\boldsymbol{j} \leq \boldsymbol{1}_{[N]}$ (i.e., $j_s \leq 1$, $\forall s \in [N]$), and similarly for $\boldsymbol{k}$.
Furthermore, the only nonvanishing terms are such that $|\boldsymbol{i}|= |\boldsymbol{n}|$ and $|\boldsymbol{j}|= |\boldsymbol{k}|$ since particle numbers are conserved in the linear interferometer.
This concludes the proof of Theorem~\ref{theo:RecNmain}.

\subsection{Proof of Theorem~\ref{theo:recNGen}}

Consider some $L \geq 2N$ and $\varepsilon \leq 1/||A||_2$ ($L$ and $\varepsilon$ always exist since $N$ is finite). From Lemma~29 in Ref.~\onlinecite{bosonsampling}, there exists a unitary matrix $V \in \mathbb{C}^{L \times L}$ that contains $\varepsilon A$ as a submatrix. Consider such a matrix $V$ constructed from $A$. Define the vectors $\boldsymbol{p}, \boldsymbol{q} \in \mathbb{N}_+^L$, as follows: if $\alpha$ ($\beta$) is the vector corresponding to the column (row) indices of $V$ that correspond to $\varepsilon A$, fill the elements with indices belonging to $\alpha$ ($\beta$) in vector $\boldsymbol{p}$ ($\boldsymbol{q}$) with the elements of $\boldsymbol{i}$ ($\boldsymbol{n}$), and put $0$'s in the rest of the elements. Since  $|\boldsymbol{i}| = |\boldsymbol{n}|$, we clearly have  $|\boldsymbol{p}| = |\boldsymbol{q}|$. Therefore, $V$ satisfies Eq.~\eqref{eq:RecNmain3}, i.e., we have
\begin{equation} \label{eq:RecNapp4}
    \sum_{\substack{\boldsymbol{r}, \boldsymbol{s} \in \mathbb{N}_+^L \\ |\boldsymbol{r}| = |\boldsymbol{s}|}}^{\boldsymbol{r} \leq \boldsymbol{p}, \boldsymbol{s} \leq \boldsymbol{q}} (-1)^{|\boldsymbol{r}|} \, \frac{| \det (V_{\boldsymbol{q}-\boldsymbol{s},\boldsymbol{p}-\boldsymbol{r}}) |^2 \, | \per (V_{\boldsymbol{s},\boldsymbol{r}}) |^2 } {(\boldsymbol{q}-\boldsymbol{s})! \, (\boldsymbol{p}-\boldsymbol{r})! \, \boldsymbol{s}! \, \boldsymbol{r}!} = 0.
\end{equation}
Now, define the vector $\boldsymbol{j} \in \mathbb{N}_+^N$ to be the subvector of $\boldsymbol{r}$ corresponding to indices in $\alpha$. Similarly, define the vector $\boldsymbol{k} \in \mathbb{N}_+^N$ to be the subvector of $\boldsymbol{s}$ corresponding to indices in $\beta$. It is clear that $\boldsymbol{r} \leq \boldsymbol{p}$ is equivalent to saying that $\boldsymbol{j} \leq \boldsymbol{i}$, while $\boldsymbol{s} \leq \boldsymbol{q}$ is equivalent to $\boldsymbol{k} \leq \boldsymbol{n}$.
It is also clear that $V_{\boldsymbol{q}-\boldsymbol{s},\boldsymbol{p}-\boldsymbol{r}} = \varepsilon A_{\boldsymbol{n}-\boldsymbol{k},\boldsymbol{i}-\boldsymbol{j}}$ and $V_{\boldsymbol{s},\boldsymbol{r}} = \varepsilon A_{\boldsymbol{k},\boldsymbol{j}}$.
From this, $\det (V_{\boldsymbol{q}-\boldsymbol{s},\boldsymbol{p}-\boldsymbol{r}}) = \varepsilon^{|\boldsymbol{i}-\boldsymbol{j}|} \det (A_{\boldsymbol{n}-\boldsymbol{k},\boldsymbol{i}-\boldsymbol{j}})$, while $\per (V_{\boldsymbol{s},\boldsymbol{r}}) = \varepsilon^{|\boldsymbol{j}|} \per (A_{\boldsymbol{k},\boldsymbol{j}})$. Since $\boldsymbol{i}-\boldsymbol{j} \geq 0$, $|\boldsymbol{i}-\boldsymbol{j}| + |\boldsymbol{j}| = |\boldsymbol{i}|$.
We can therefore rewrite Eq.~\eqref{eq:RecNapp4} as
\begin{equation}
    \sum_{\substack{\boldsymbol{j}, \boldsymbol{k} \in \mathbb{N}_+^N \\ |\boldsymbol{j}| = |\boldsymbol{k}|}}^{\boldsymbol{j} \leq \boldsymbol{i}, \boldsymbol{k} \leq \boldsymbol{n}} (-1)^{|\boldsymbol{j}|} \, \frac{\varepsilon^{2 |\boldsymbol{i}|} | \det (A_{\boldsymbol{n}-\boldsymbol{k},\boldsymbol{i}-\boldsymbol{j}}) |^2 \, | \per (A_{\boldsymbol{k},\boldsymbol{j}}) |^2 } {(\boldsymbol{n}-\boldsymbol{k})! \, (\boldsymbol{i}-\boldsymbol{j})! \, \boldsymbol{k}! \, \boldsymbol{j}!} = 0.
\end{equation}
Dividing both sides of the identity by $\varepsilon^{2 |\boldsymbol{i}|}$ ends the proof.

%%%%%%%%%%%%%%%%%%%%%%%%
%%%%%%%%%%%%%%%%%%%%%%%%
%%%%%% CONCLUSION %%%%%%
%%%%%%%%%%%%%%%%%%%%%%%%
%%%%%%%%%%%%%%%%%%%%%%%%

\section{Conclusions and discussion}

By analytically expressing the generating function of bosonic probabilities, we have found a fundamental relation expressing a duality between the bosonic and fermionic particle statistics in an arbitrary linear interferometer (our Theorem~\ref{theo:RecNmain}). This relation takes the form of a linear equation combining bosonic and fermionic transition probabilities, which was \textit{a priori} implausible since quantum interferences inherently involve linear combinations of probability amplitudes (not of their squares). Restricting to one particle  per mode at most, this relation constrains the sum (difference) of bosonic and fermionic probabilities for an even (odd) particle number. For a single particle, the difference simply vanishes (reflecting the fact that particle statistics plays no role), whereas for two particles, the sum is equal to twice the classical probability (reflecting a boson--fermion complementarity). This result  generalizes to an arbitrary number of particles and happens to stem from a heretofore unknown mathematical identity between the squared moduli of determinants and permanents of arbitrary complex matrices
(our Theorem~\ref{theo:recNGen}). It would be interesting to explore the consequence of Theorems~\ref{theo:RecNmain} and~\ref{theo:recNGen} on the complexity of simulating linear optical circuits. The boson sampling paradigm  derives from the computational hardness of computing permanents, hence the connection with (easily computable) determinants as implied by Theorems~\ref{theo:RecNmain} and~\ref{theo:recNGen} may have intriguing ramifications.
\par

Finally, we remark that our result bears some similarities with the so-called \emph{quantum} MacMahon theorem~\cite{Q_MacMahon1,Q_MacMahon2,Q_MacMahon3}. The latter expresses a relation between some suitable quantum generalizations of the determinant and permanent, which can be interpreted as a sort of boson-fermion correspondence. Nevertheless, the quantum MacMahon theorem seems quite different from our Theorem~\ref{theo:RecNmain}, which fundamentally connects the statistics of bosons and fermions in an arbitrary interferometer.
This uncovered complementarity between bosons and fermions suggests a possible explanation in terms of a supersymmetry algebra, which is worth further investigation.

Let us conclude with the question of whether the boson-fermion complementarity remains valid in a more general context, in particular one that better describes a realistic situation. When it comes to the experimental manipulation of bosonic systems, in particular photonic ones, particles going through interferometers such as the ones described in the present work unavoidably exhibit slight discrepancies in their internal degrees of freedom, such as polarization or time. As a consequence, the particles are effectively not fully indistinguishable, but rather partially distinguishable. In view of the importance of interferometric experiments in recent years, partial distinguishability is worth studying. A natural question is then whether the boson-fermion complementarity holds for partially distinguishable particles. It appears that a statement analogous to that of Theorem~\ref{theo:RecNmain} holds in that case, as hinted at by preliminary results~\cite{Marco}.

\begin{acknowledgments}
We are grateful to Leonardo Novo for insightful discussions.
M.G.J. acknowledges support by the Fonds de la Recherche Scientifique – FNRS (Belgium).
N.J.C. acknowledges support by the Fonds de la Recherche Scientifique – FNRS (Belgium) under project CHEQS within the Excellence of Science (EOS) program.
\end{acknowledgments}

\section*{Data Availability Statement}

Data sharing is not applicable to this article as no new data were created or analyzed in this study.

\bigskip

% Bibliography
\bibliography{bosfer}

\end{document}